\documentclass[aip,
jmp, 
graphicx, amsmath, amssymb, amsthm,
draft,
author-year,
preprint 
]{revtex4-1}

\usepackage[colorlinks=true,linkcolor=black,citecolor=black,urlcolor=blue]{hyperref}

\usepackage{amsthm,amsfonts}
\usepackage{theoremref}

\newtheorem{theorem}{Theorem}[section]

\newtheorem{lemma}{Lemma}[section]

\newtheorem{corollary}[theorem]{Corollary}
\newtheorem{example}{Example}[section]
\newtheorem{remark}{Remark}[section]
\newtheorem{definition}{Definition}[section]

\usepackage{subcaption}

\usepackage{textcomp} 
\usepackage{stmaryrd} 

\usepackage{mathtools}

\usepackage[capitalize,nameinlink]{cleveref} 
\crefformat{definition}{Definition~#2#1#3}
\crefformat{equation}{Equation~(#2#1#3)}
\crefformat{lemma}{Lemma~#2#1#3}
\crefformat{theorem}{Theorem~#2#1#3}
\crefformat{section}{Section~#2#1#3}

\usepackage{thmtools,thm-restate}

\usepackage{url}

\usepackage{mathpartir}

\usepackage{booktabs}
\usepackage[raggedrightboxes]{ragged2e}
\usepackage{multirow}

\usepackage{xcolor}
\definecolor{ltblue}{rgb}{0,0.4,0.4}
\definecolor{dkblue}{rgb}{0,0.1,0.6}
\definecolor{dkgreen}{rgb}{0,0.35,0}
\definecolor{dkviolet}{rgb}{0.3,0,0.5}
\definecolor{dkred}{rgb}{0.5,0,0}

\usepackage{listings}
\usepackage{tikz} 
\usetikzlibrary{cd}
\usetikzlibrary{positioning}
\usetikzlibrary{shapes.multipart}
\usetikzlibrary{%
  arrows,%
  shapes.misc,
  shapes.arrows,%
  shapes.callouts,
  chains,%
  matrix,%
  positioning,
  scopes,%
  decorations.pathmorphing,
  decorations.text,
  shadows%
}
\usepackage{quiver}
\usepackage{qcircuit}

\usepackage{textcomp} 

\usepackage{xparse}
\NewDocumentCommand{\optionalParens}{s m m}{
    \IfBooleanTF{#2}{\left(#3\right)}{\IfBooleanTF{#1}{~#3}{#3}}
}
\NewDocumentCommand{\apply}{m s O{} m}{
    #1#3 \optionalParens*{#2}{#4}
}

\NewDocumentCommand{\applytwo}{m O{} s m s m}{ 
   #1#2~\optionalParens{#3}{#4}~\optionalParens{#5}{#6}
}

\usepackage{xspace}

\newcommand{\eg}{\emph{e.g.,}\xspace}
\newcommand{\ie}{\emph{i.e.,}\xspace}

\usepackage{cmll} 
\newcommand{\lolli}{\multimap}

\usepackage{braket}

\newcommand{\ZZ}{\mathbb{Z}}

\usepackage{enumerate}

\NewDocumentCommand{\Zd}{st'}{\IfBooleanTF{#1}{\tfrac12\mathbb{Z}_{d'}}
                               { \IfBooleanTF{#2}{\mathbb{Z}_{d'}}
                               { \mathbb{Z}_d
                               }}}
\RenewDocumentCommand{\ZZ}{st'}{\IfBooleanTF{#1}{\tfrac12\mathbb{Z}_{4}}
                               { \IfBooleanTF{#2}{\mathbb{Z}_{4}}
                               { \mathbb{Z}_2
                               }}}

\NewDocumentCommand{\pcl}{s}{\IfBooleanTF{#1}{\mathit{PCl}}{\mathbf{PCl}}}
\NewDocumentCommand{\pauligroup}{s}{\IfBooleanTF{#1}{\mathit{P}}{\mathbf{P}}}
\newcommand{\spgroup}{\mathrm{Sp}}

\newcommand{\cat}[1]{\mathcal{#1}}
\newcommand{\PCat}{\cat{P}_c}

\newcommand{\lambdaP}{\lambda^{\PCat}}

\renewcommand{\textsf}{\text}

\newcommand{\cprod}{\star}
\newcommand{\ptensor}{\bullet}

\newcommand{\qudit}{\text{qu\emph{d}it}\xspace}
\newcommand{\qudits}{\text{qu\emph{d}its}\xspace}

\NewDocumentCommand{\CQSTATE}{O{}}{\textsf{CQ}^{#1}}

\NewDocumentCommand\interp{sm}
        {\IfBooleanTF{#1}
                {\conjugate{\left\llbracket #2 \right\rrrbracket}}
                {\left\llbracket #2 \right\rrbracket}
        }

\NewDocumentCommand\commute{smm}
        {#2 \IfBooleanT{#1}{\not}\upmodels #3}

\NewDocumentCommand{\symplecticform}{O{}mm}{\omega_{#1}\left(#2,#3\right)}
\NewDocumentCommand{\SYMPLECTICFORM}{O{}}{\overline{\omega}_{#1}}

\NewDocumentCommand{\discard}{O{}m}{\texttt{discard}_{#1}(#2)}
\NewDocumentCommand{\dup}{O{}m}{\texttt{dup}_{#1}(#2)}

\NewDocumentCommand{\norm}{O{} m}{\left\lVert #2 \right\rVert_{#1}}

\NewDocumentCommand{\pseudoinverse}{O{} m}{\textsf{pinv}_{#1}(#2)}
\NewDocumentCommand{\pseudomuinverse}{O{} m}{\textsf{minv}_{#1}(#2)}

\NewDocumentCommand{\psiof}{sm}{
    \IfBooleanTF{#1}{(#2)^{\psi}}
                    {{#2}^{\psi}}
}

\begin{document}


\title{Condensed Encodings of Projective Clifford Operations in Arbitrary Dimension} 


\author{Sam Winnick}%
\email[Corresponding author: ]{samwinnick1@hotmail.com}
\affiliation{%
  {Institute for Quantum Computing and Department of Combinatorics \& Optimization, University of Waterloo. Waterloo, Ontario}
}

\author{Jennifer Paykin}
\affiliation{%
  {Intel Corporation. Hilsboro, Oregon}
}


\date{\today}

\begin{abstract}
    We provide a careful analysis of a structure theorem for the $n$-\qudit projective Clifford group and various encoding schemes for its elements. In particular, we derive formulas for evaluation, composition, and inversion. Our results apply to all integers $d\geq2$, most notably the case where $d$ is even.
\end{abstract}


\maketitle 

\section{Introduction}

In quantum computing, implementations of stabilizer simulation~\citep{aaronson2004,gidney2021stim,smith2023clifford}, circuit optimizers~\citep{paykin2023pcoast,schneider2023sat}, formal verification~\citep{fang2024symbolic} and more rely on efficient algebraic \emph{encodings} of Pauli and Clifford operators---
finite descriptions often consisting of discrete-valued vectors and matrices. For example, a Pauli tableau~\citep{aaronson2004} encodes a projective Clifford operator (an equivalence class of Clifford operators modulo phase) by recording its conjugation action on each standard non-central generator of the Pauli group. Such encodings are efficient as they represent an $n$-qubit projective Clifford in polynomial space, or equivalently, simulate an $n$-qubit stabilizer state in polynomial time. 

Encodings of Paulis and Cliffords have also been generalized to \qudits---quantum systems of arbitrary dimension $d$---which are of interest for both theoretical and practical applications~\citep{App05,hostens2005stabilizer,looi2008quantum,gurevich2012weil,deBeadrap2013linearized,raussendorf2023role,brandl2024efficient}.  
These encodings are well-behaved for odd dimensions $d$, but are more complicated for even dimensions. For example, in the relation $Y=\tau XZ$ satisfied by the \qudit Pauli operators, $\tau$ is a primitive $d$th root of unity when $d$ is odd and a primitive $2d$th root of unity when $d$ is even.
Letting $d'$ be the order of $\tau$~\citep{App05}, the elements of the Pauli group $\tau^r X^x Z^z$ are in bijective correspondence with tuples $(r,(x,z)) \in \Zd' \times \Zd^{2n}$.
Encodings of \qudit projective Cliffords found in the literature typically incorporate the ring $\Zd'$ in some way when considering arbitrary dimension $d$.

Working within $\Zd'$ is not always straightforward however, and much of the work in the literature treats the even and odd cases separately.
For example, when $d$ is even, certain calculations \eg about Wigner functions~\citep{mari2012positive,gross2006hudson} break down because $2$ lacks an inverse mod $d$. 

This work develops the theory of \emph{condensed encodings}, an encoding scheme that uses only transformations over $\Zd$ and not $\Zd'$.
It is well-known that the projective Clifford group $\pcl_{d,n}$ is an extension of the symplectic group $\spgroup(\Zd^{2n})$ by the group of linear functionals
$V^*$~\citep{bolt1961clifford}, from which Lagrange's Theorem implies the existence of a bijection between $\pcl_{d,n}$ and $(\Zd^{2n})^*\times \spgroup(\Zd^{2n})$. Choosing such a bijection is equivalent to choosing a section $\pcl_{d,n}\leftarrow\spgroup(\Zd^{2n})$, and the standard basis for $\Zd^{2n}$ provides a canonical choice, which we call the \emph{condensed encoding}.

Condensed encodings appear implicitly 
in the literature, notably by \citet{gurevich2012weil}, but as far as we are aware, relevant formulas have not previously been made explicit.
\citet{raussendorf2023role} noted that the group law for such encoded projective Clifford operations depends on the group $2$-cocycle associated with the extension~\citep{conrad2018group}, but did not provide explicit formulas for the cocycle.
In this paper we develop a concrete formula for a related \emph{phase correction} function, which enables us to evaluate a projective Clifford in terms of its condensed encoding, \ie given an encoding, compute the corresponding Clifford action on the Pauli group. The phase correction function also gives rise to formulas for composition and inversion, all of which are necessary for applications of condensed encodings to stabilizer simulation and error correction.

Our technical treatment of encodings makes several contributions: (1) We introduce novel notational tricks for dealing with the even and odd cases uniformly. (2) We offer a careful analysis of encodings of projective Clifford operations involving the extended phase space in the even case using homomorphic relations and symplectic lifts. (3) We provide explicit formulas for condensed encodings, which have so far only been noted implicitly. (4) We introduce the condensed product $\star$, which plays a role similar enough to ordinary matrix multiplication to serve as a preferable substitute in applications, and features better compatibility with condensed encodings.

\textbf{Uniform reasoning in even and odd dimensions.}
Although the condensed encoding does not use the ring $\Zd'$ itself, the proofs of its correctness do. In order to improve reasoning about the odd and even cases in a unified manner, we introduce the $\Zd$-module $\frac12\Zd'$ obtained by adjoining $\Zd$ with an element $\frac12$ satisfying $\frac12+\frac12=1$. This is achieved using the isomorphisms $2:\frac12\Zd'\stackrel\sim\to\Zd'$ and $\frac12:\Zd'\stackrel\sim\to\frac12\Zd'$ in place of multiplication by $2$ or $\frac12$, respectively. We apply this ``fake $\frac12$'' technique in our analysis of the projective Clifford group and are explicit about the domain of each variable in $\Zd$, $\Zd'$, or $\frac12\Zd'$. We also use explicit notation to keep track of reduction (mod $d$) $\overline{\,\cdot\,}:\Zd'\to\Zd$ and inclusion $\underline{\,\cdot\,}:\Zd\to\Zd'$.

\textbf{Other encodings of \qudit Cliffords.}
In the process of developing the theory of condensed encodings, we recover other encodings of \qudit projective Cliffords found previously in the literature. We first discuss encodings involving the displacement operators $D_{x,z}=X^xZ^z$, which we refer to as $D$-encodings~\citep{hostens2005stabilizer}. The drawbacks of these $D$-encodings motivate encodings the use of Weyl operators $\Delta_{x,z}=\tau^{x\cdot z}X^xZ^z$, which we refer to as $\Delta$-encodings~\citep{App05,deBeadrap2013linearized}. We then explore various $\Delta$-encodings using the extended phase space $\Zd'^{2n}$, which we call \textit{extended encodings}, before finally arriving at the \textit{condensed encodings}, which involve the original phase space $\Zd^{2n}$. 

The most closely related encodings found in the literature are the $\Delta$-encodings by \citet{App05} and \citet{deBeadrap2013linearized}. The former gives a detailed account of the single \qudit case, and relies on factoring matrices in $\spgroup(\Zd'^2)=\mathrm{SL}(\Zd'^2)$, while the latter states the result for multiple \qudits but does not provide a fully detailed proof.
Our analysis of $\Delta$-encodings of projective Clifford operations is different than Appleby's in that, rather than working with the special linear group, we rely on basic properties of homomorphic relations as well as the fact that a symplectomorphism mod $d$ can be \emph{lifted} to a symplectomorphism mod $d'$ by a theorem of \citet{newmansmart}.

\textbf{The condensed product.}
\citet{raussendorf2023role} noted that when $d$ is even, elements of the Pauli group of order at most $d$ play a special role. We introduce a non-associative operation $\star$ that preserves the subset $\mathbf Q_{d,n}$ of the Pauli group consisting of such elements, generalizing the Hermitian product used by \citet{paykin2023pcoast}. This operation, as well as the condensed encodings themselves, are motivated by the authors' recent work to construct a programming language for \qudit projective Cliffords called $\lambdaP$~\citep{paykin2024qudit}. In $\lambdaP$, types correspond to elements of $\mathbf Q_{d,n}$ and terms correspond to projective Cliffords via a Curry-Howard correspondence between $\lambdaP$ and a categorical model of condensed encodings. The formulas we derive in this paper therefore arise both in the categorical semantics as well as in the operational semantics of $\lambdaP$ via the $\beta$-reduction rules. 

\textbf{Outline.}
The paper proceeds as follows: \cref{sec-paulicliff} reviews the Pauli and projective Clifford groups for arbitrary $d$. \cref{sec-denc,sec-deltaenc,sec-lemmas,sec-neccsuff} cover a series of encodings, refined through a series of steps to finally arrive at the desired condensed encoding scheme. These different encodings are summarized in \cref{fig:Encodings}. At the heart of the proof when $d$ is even is the nontrivial fact that a symplectic form on $\Zd^{2n}$ can be lifted to a symplectic form on $\mathbb Z_{2d}^{2n}$. Finally, \cref{sec-recipe,sec-odot} work out explicit formulas and identities involving condensed encodings, and  \cref{sec-examples} collects important examples.
    



\begin{figure*}
    \centering
    \begin{ruledtabular}
    \begin{tabular}{p{0.3\textwidth} p{0.2\textwidth} p{0.45\textwidth}}
        $D$-encoding $(r,\psi)$

        (\cref{sec-denc})
            & {$\!\begin{aligned}[t]
                r &: V \to \tfrac12 R' \\
                \psi &: V \to V \end{aligned}$}
            &
            $D_v \mapsto \zeta^{r v} D_{\psi v}$ a projective Clifford 
            
            (center-fixing group automorphism)
        \\ \hline
        pre-$\Delta$-encoding $(\lambda,\phi)$

        (\cref{sec-lemmas})
            & {$\!\begin{aligned}[t]
                \lambda &: V' \to \tfrac12 R' \\
                \phi &: V' \to V'
                \end{aligned}$}
            &
            $\Delta_v \mapsto \zeta^{\lambda v} \Delta_{\phi v}$ a center-respecting total relation
        \\ \hline
        $\Delta$-encoding $(\lambda,\phi)$
        
        (\cref{sec-deltaenc,sec-lemmas})
            & {$\!\begin{aligned}[t]
                \lambda &: V' \to R \leq \tfrac12 R' \\
                \phi &: V' \to V'
                \end{aligned}$}
            & 
            $\Delta_v \mapsto \zeta^{\lambda v} \Delta_{\phi v}$ a projective Clifford
        \\ \hline
        Condensed encoding $(\mu,\psi)$

        (\cref{sec-recipe})
            & {$\!\begin{aligned}[t]
                \mu &: V \to R \\
                \psi &: V \to V
                \end{aligned}$}
            & $\Delta_{\underline b} \mapsto \zeta^{\mu b} \Delta_{\underline{\psi b}}$ a projective Clifford
            
            given by action on basis elements $b \in V$
    \end{tabular}
    \end{ruledtabular}
    \caption{Variants of encodings of center-respecting relations and projective Cliffords.}
    \label{fig:Encodings}
\end{figure*}

\section{The Pauli group and the projective Clifford group}
\label{sec-paulicliff}

Let $d>1$ be an integer and let $d'=2d$ if $d$ is even and $d'=d$ if $d$ is odd. Let $\tau$ be a primitive $d'$th root of unity, and set $\zeta=\tau^2$. Then $\zeta$ has order $d$, regardless of whether $d$ is even or odd. Next, set $R=\mathbb Z_d$ and $R'=\mathbb Z_{d'}$, and define $\frac12R'$ by adjoining half-elements to $R$ if necessary. Explicitly, an abelian group presentation for $\frac12R'$ is:
\begin{align*}
    \frac12R' = \left.\bigg(1,\frac12\,\right\rvert\,\underbrace{1+\cdots+1}_d=0\textrm{ and }\frac12+\frac12=1\bigg)
\end{align*}
If $d$ is even, $2$ has no multiplicative inverse, so $\tfrac12$ is a new element; if $d$ is odd, $\tfrac12=\tfrac{d+1}2\in R$, so $R=R'=\tfrac12R'$. In either case we have mutually inverse isomorphisms of $R$-modules $2:\frac12R'\stackrel\sim\to R'$ and $\frac12:R'\stackrel\sim\to\frac12R'$ defined in the obvious way. Our reason for introducing $\frac12R'$ is so that we can interpret $\zeta^r$ as $\tau^{2r}$ where $2r\in R'$ is the image of $r\in\frac12R'$ under the isomorphism $2$. Regardless of parity, $\frac12R'=\{t\in R'\mid t\in R'\}$, thus justifying the notation. 

\begin{remark}
The $R$-module $\frac12R'$ should not be confused with the ring $R[\frac12]$ obtained by adjoining $\frac12$ in the sense of ring theory. In particular, $\frac12R'$ is not closed under multiplication.
\end{remark}

We write $X$ and $Z$ for the single-\qudit Pauli matrices. In the computational basis, their action is:
\begin{align*}
    X\ket r=\ket{r+1} && Z\ket r=\zeta^r\ket r && r\in R
\end{align*}
We also set $Y=\tau XZ$. For each non-negative integer $n$, the $n$-\qudit Pauli matrices are
\begin{align*}
    X^q=X^{q_1}\otimes\cdots\otimes Z^{q_n} && Z^p=Z^{p_1}\otimes\cdots\otimes Z^{p_n} && q,p\in R^n
\end{align*}
where $\otimes$ is the Kroncker product. Let $V=R^n \oplus R^n = R^{2n}$. For each $v=(q,p)\in V$, we write $D_v=X^qZ^p$. 



The Pauli group, generated by $X$, $Y$, and $Z$, has order $d' d^{2n}$:
\begin{align*}
    \pauligroup_{d,n} &= \langle X,Y,Z\rangle = \langle \tau,X,Z\rangle = \langle\tau,D_v\,|\,v\in V\rangle \\
    &= \{\tau^tD_v\,|\,t\in R',\,v\in V\} = \{\zeta^tD_v\,|\,t\in\tfrac12R',\,v\in V\}
\end{align*}
It gives rise to the projective Clifford group as follows:
\small \begin{align*}
    \pcl_{d,n}&=\{\gamma\in\mathrm{Aut\,P'}_{d,n}\,|\,\gamma z=z\,\forall\,z\in\mathrm Z(\pauligroup_{d,n})\} \cong \{[U]\in\mathrm{PU}(d^n)\,|\,UPU^{-1}\in\pauligroup_{d,n}\,\forall\,P\in\pauligroup_{d,n}\}
\end{align*} \normalsize
Here $\mathrm Z(G)=\{z\in G\,|\,gz=zg\,\forall\,g\in G\}$ is the center. The isomorphism $\cong$ above identifies each automorphism of the form $g\mapsto UgU^{-1}$ with the equivalence class $[U]$ of unitaries up to $\mathrm{U}(1)$ phase. Either the Stone-von Neumann Theorem or the Skolem-Noether Theorem may be used to establish the nontrivial direction of these isomorphisms. In the theory, we characterize $\pcl_{d,n}$ as the group of center-fixing automorphisms, as on the left-hand-side of the $\cong$, though examples are typically given as on the right-hand-side, since they are given by unitary matrices, eg. the matrix $S=(\begin{smallmatrix}
    1\\&i
\end{smallmatrix})$.


\begin{remark}
    Some applications consider the Pauli group generated by $X$ and $Z$ rather than $X$, $Y$, and $Z$; we refer to this variant as the \emph{small} Pauli group. In this work we favor $\pauligroup_{d,n}=\langle X,Y,Z\rangle$ for several reasons. First, the projective Clifford group arising from $\langle X,Y,Z\rangle$ contains the projective $S$ gate, \textit{i.e.} conjugation by the matrix $S=(\begin{smallmatrix}1\\&i\end{smallmatrix})$, whereas the version arising from $\langle X,Z\rangle$ does not. Second, $\pcl_{d,n}$ is equivalently defined in terms of the Heisenberg-Weyl group $\mathbf{HW}_{d,n}=\langle e^{i\theta},X,Z\,|\,\theta\in\mathbb R\rangle$ in place of $\pauligroup_{d,n}$, so taking $\langle\tau\rangle$ as the center is sufficient. The structure of both the small and ``large'' versions of the projective Clifford group is worked out by \citet{bolt1961clifford}.
\end{remark}

\section{\(D\)-encodings}\label{sec-denc}

In this section we investigate the properties of $D_v=X^qZ^p$ introduced in \cref{sec-paulicliff}. Later, in \cref{sec-deltaenc}, we will consider an alternative $\Delta_v$. We refer to $D_v$ as \textit{displacement operators} and to $\Delta_v$ as \textit{Weyl operators}. The advantage of $D_v$ is that each $g\in\pauligroup_{d,n}$ has the form $\zeta^tD_v$ for unique $t\in\frac12R'$ and $v\in V$. With $(t,v)$ standing for $\zeta^tD_v$, the group law for $\pauligroup*_{d,n}'$ reads:
\begin{align}\label{dencpaulilaw}
    (t_1,v_1)(t_2,v_2) = (t_1+t_2+{p_1\cdot q_2},v_1+v_2)
\end{align}
where $v_i=(q_i,p_i)$ and $p_1\cdot q_2$ is put through the \textit{inclusion} homomorphism $R\to\frac12R'$. In this sense, the dot product $\cdot:V\otimes V\to R$ encodes the group law of $\pauligroup_{d,n}$. Likewise, the symplectic form $\omega:V\otimes V\to R \leq \Zd*$ given by $\omega(v_1,v_2)=p_1\cdot q_2-p_2\cdot q_1$ encodes the group commutator of $\pauligroup_{d,n}$:
\begin{align}\label{dcommutator}
    (t_1,v_1)(t_2,v_2)(t_1,v_1)^{-1}(t_2,v_2)^{-1}=({\omega(v_1,v_2)},0)
\end{align}

The existence and uniqueness of such \textit{$D$-encodings} of elements of $\pauligroup_{d,n}$ (\textit{i.e.} as pairs $(t,v)$ or expressions $\zeta^tD_v$) likewise implies that each $\gamma\in\pcl_{d,n}$ has a unique \textit{$D$-encoding} by a pair of functions $(r,\psi)$ of the forms $r:V\to\frac12R'$ and $\psi:V\to V$.\footnote{This $D$-encoding of Cliffords is very similar to the one presented by \citet{hostens2005stabilizer}, except their extended phase space falls entirely within $\mathbb{Z}_{2d}$ rather than $\Zd*$.} By this we mean that $\gamma(D_v)=\zeta^{rv}D_{\psi v}$ for each $v\in V$. As with $D$-encodings of Paulis, we may express $D$-encodings of projective Cliffords either using the ordered pair notation or using the original operator notation. 

Not every pair $(r,\psi)$ of functions of these types is the $D$-encoding of some projective Clifford. For instance, using the fact that $\gamma$ is a homomorphism and the uniqueness of $D$-encodings, we find that $\psi\in \spgroup(V,\omega)$ by extracting the exponent of the following:
\begin{align}\label{psihomo}
    \zeta^{\omega(v_1,v_2)}=\gamma(D_{v_1}D_{v_2}D_{v_1}^{-1}D_{v_2}^{-1})=\gamma(D_{v_1})\gamma(D_{v_2})\gamma(D_{v_1})^{-1}\gamma(D_{v_2})^{-1}=\zeta^{\omega(\psi v_1,\psi v_2)}
\end{align}

A major disadvantage of $D$-encodings of projective Cliffords is that the structure of $r$ is unclear. Consider, for example, the single-qubit $S$ gate. In this case, $\gamma$ is conjugation by $S=(\begin{smallmatrix}1\\&i\end{smallmatrix})$, where $SXS^{-1}=Y$ and $SZS^{-1}=Z$. Then $r:\mathbb Z_2^2\to\frac12R'=\{0,\frac12,1,\frac32\}$ is given by $r(0,0)=r(0,1)=0$ and $r(1,0)=r(1,1)=\frac12$. Notably, $r$ is not a homomorphism since $r(1,0)+r(1,0) = \tfrac12 + \tfrac12 \neq 0 = r(0,0) = r((1,0) + (1,0))$. 

Even though $r$ might not be a homomorphism, its action can be determined from $\psi$ on every standard basis vector $b \in V$. To see why, consider that every $v \in V$ is either $0$ or has the form $u + b$ for some standard basis vector $b$ and some $u \in V$ smaller than $v$ (in the obvious sense). Furthermore, from \cref{dencpaulilaw} we derive the following:
\begin{align}\label{denc-scalarprop}
    rv_1+rv_2=r(v_1+v_2)+{c(v_1,v_2)-c(\psi v_1,\psi v_2)}
\end{align}
where $c((q_1,p_1),(q_2,p_2))=p_1\cdot q_2$. Note that on the right-hand-side of \cref{denc-scalarprop}, the $c$ function is valued in $R\leq\tfrac12R'$. The inclusion implicit here and in \cref{dcommutator} is a homomorphism, unlike (and not to be confused with) the subset inclusion $R\subseteq R'$, which is not (when $d$ is even).

\section{\(\Delta\)-encodings}\label{sec-deltaenc}

There is another approach that allows for better structured encodings. Initially, this comes at the cost of uniqueness, but we later address this by introducing \textit{condensed encodings}. The first goal is to replace the bilinear form $c$ in \cref{denc-scalarprop} by the symplectic form $\omega$, or something similar, since $\psi$ being a symplectomorphism would then imply the new function replacing $r$ is a homomorphism. This is an oversimplification of the end result, which involves a tremendous amount of technicalities, but does ultimately lead to the condensed encoding result, \cref{thm-denseenc}. This approach involves using using the \textit{Weyl operators} $\Delta_v$ in place of the displacement operators discussed in \cref{sec-denc}, appearing for instance in \citet{bolt1961clifford} and \citet{AFMY17}. 

Set $V'=(R')^{2n}$, and given $v=(q,p)\in V'$, let $\Delta_v=\tau^{q\cdot p}D_{\overline v}$ where the overline $\overline{v}$ indicates the reduction mod $d$ homomorphism $V'\to V$. The notation $\zeta^t\Delta_v$ with $t\in\frac12R'$ and $v\in V'$ \emph{non-uniquely} expresses an element of $\pauligroup_{d,n}$; for instance, for a single qubit we have $\Delta_{1,0}=X=\Delta_{3,0}$ and $\Delta_{1,1}=Y\neq-Y=\Delta_{3,1}$. Nevertheless, writing $(t,v)$ in place of the \emph{expression} $\zeta^t\Delta_v$, the group law for $\pauligroup_{d,n}$ may be expressed non-uniquely as:
\begin{align}
    (t_1,v_1)(t_2,v_2)=(t_1+t_2+\tfrac12\omega'(v_1,v_2),v_1+v_2)
\end{align}
Here, $\omega':V'\otimes V'\to R'$ is the \textit{standard extended symplectic form}  $\omega'((q_1,p_1),(q_2,p_2))=p_1\cdot q_2-p_2\cdot q_1$, and $\tfrac12$ denotes the inverse of the isomorphism $2:\tfrac12R'\stackrel\sim\to R'$. Observe that $\Delta_v^t=\Delta_{tv}$, and in particular, $\Delta_v^d=\Delta_{dv}=\tau^{d^2q\cdot p}D_{dv}=I$ for all $v'\in V$. Thus $\Delta_v$ has order dividing $d$, in contrast to the operators $D_v$, where the single qubit operator $D_{1,1}=XZ$ has order $d'=4$.


Having noted the complicated behavior of $\Delta_v$, we introduce the gadget $\langle\epsilon\rangle_{v_1,v_2}$ to keep track of when negative signs are introduced. Let $v_1=(q_1,p_1),v_2=(q_2,p_2)\in V'$ with $v_2=v_1+d\epsilon$ for some $\epsilon=(\epsilon_{\mathrm q},\epsilon_{\mathrm p})\in\mathbb Z_{d'/d}^{2n}$. Then $q_2\cdot p_2=q_1\cdot p_1+d\langle\epsilon\rangle_{v_1,v_2}$, where we define:
\begin{align}\label{eqn-computeeps}
    \langle\epsilon\rangle_{v_1,v_2}=\overline q_1\cdot\epsilon_{\mathrm p}+\overline p_1\cdot\epsilon_{\mathrm q}=\overline q_2\cdot\epsilon_p+\overline p_2\cdot\epsilon_q\in\mathbb Z_{d'/d}=\begin{cases}\mathbb Z_2&:d\textrm{ even}\\\textrm{trivial}&:d\textrm{ odd}\end{cases}
\end{align}

It follows that $\Delta_{v_2}=(-1)^{\langle\epsilon\rangle_{v_1,v_2}}\Delta_{v_1}$. We will make frequent use of this $\langle\epsilon\rangle$ notation. Given $v_1$ and $\epsilon$, we may write $\langle\epsilon\rangle_{v_1}$ as shorthand for $\langle\epsilon\rangle_{v_1,v_1+d\epsilon}$, and given $v_1$ and $v_2$ with equal reduction mod $d$, we may implicitly define $\epsilon$ by writing $\langle\epsilon\rangle_{v_1,v_2}$. If the context is clear, we may even simply write $\langle\epsilon\rangle$. This lets us account for the correction factor without overemphasizing the details. One obvious property is that $\langle\epsilon\rangle_{v,v}=0$, and a less obvious one is that $\langle\epsilon\rangle_0=0$. The latter implies $\langle\epsilon\rangle_{v_1,v_2}=\frac1d\omega'(v_1,v_2)$ (but still only when $v_2=v_1+d\epsilon$).

We now turn our attention to projective Cliffords. Like with $D$-encodings, it is possible to consider encodings not just of Paulis, but also of projective Cliffords. Let $\gamma\in\pcl_{d,n}$ and suppose there exist functions $\lambda : V' \to \tfrac12 R'$ and $\phi : V' \to V'$ such that $\gamma(\Delta_v)=\zeta^{\lambda v} \Delta_{\phi v}$. We call such a pair $(\lambda,\phi)$ a \textit{$\Delta$-encoding} of $\gamma$. We discuss the existence of $\Delta$-encodings and related variants in the next section.

As we did for $D$-encodings of projective Cliffords, we now ask: what properties characterize $\lambda$ and $\phi$?
For one, consider that we must have $\gamma(\Delta_{v_1}\Delta_{v_2})=\gamma(\Delta_{v_1})\gamma(\Delta_{v_2})$. Then:
\begin{align*}\gamma(\Delta_{v_1}\Delta_{v_2})&=\zeta^{\lambda(v_1+v_2)+\frac12\omega'(v_1,v_2)}\Delta_{\phi(v_1+v_2)}\\
    \gamma(\Delta_{v_1})\gamma(\Delta_{v_2})&=\zeta^{\lambda v_1+\lambda v_2+\frac12\omega'(\phi v_1,\phi v_2)}\Delta_{\phi v_1+\phi v_2}
\end{align*}
It need not necessarily be the case that $\phi(v_1+v_2)=\phi v_1+\phi v_2$ since
the function taking each $v\in V'$ to $\Delta_v$ is not injective in general. However, the function taking $\overline{v}\in V$ to $\{\pm\Delta_v\}$ (where $v \in V'$ is any lift of $\overline v$) is a well-defined injective function. It follows that the function $\overline\phi:V'\to V$ obtained by reducing the output mod $d$ is $R'$-linear, which then allows us to reduce the \emph{input} mod $d$ as well. We use the following notation for the $R$-linear map obtained from $\phi$ by reducing the input and output mod $d$: 
\begin{align*}
\psi=\overline{\overline\phi}:V\to V
\end{align*}
Finally, analysis of the phase exponent yields a formula in analogy to \cref{denc-scalarprop}:
\begin{align}\label{deltaenc-scalarprop}
    \lambda v_1+\lambda v_2-\lambda(v_1+v_2)+\frac d2\langle\epsilon\rangle_{\phi(v_1+v_2),\phi v_1+\phi v_2}=\frac12\omega'(v_1,v_2)-\frac12\omega'(\phi v_1,\phi v_2)\in\frac12R'
\end{align}
%
If $\phi$ is $R'$-linear then $\langle\epsilon\rangle=0$, in which case $\lambda$ is a homomorphism if and only if $\phi\in\spgroup(V',\omega')$. If we now ``multiply'' \cref{deltaenc-scalarprop} by $2:\frac12R'\stackrel\sim\to R'$ and then reduce mod $d$, we obtain:
\begin{align}
    \overline{2\lambda v_1+2\lambda v_2-2\lambda(v_1+v_2)}=\omega(\overline v_1,\overline v_2)-\omega(\psi\overline v_1,\psi\overline v_2)\in R && \textrm{for all }v_1,v_2\in V'
\end{align}
The left-hand-side of the equation above is symmetric and the right-hand-side is alternating, hence skew-symmetric. But if $d$ is even, this is not enough to conclude that $\psi\in\spgroup(V,\omega)$.

Observe however that $\psi$ is the same map involved in the $D$-encoding of $\gamma$, and so it must be the case that $\psi\in\spgroup(V,\omega)$.
It follows that $\lambda$ is ``almost'' a homomorphism, in the sense that for each $v_1,v_2\in V'$, there exists $\alpha\in\mathbb Z_{d'/d}$ such that $\lambda v_1+\lambda v_2-\lambda(v_1+v_2)=d\alpha$.

\section{Existence and refinement of \(\Delta\)-encodings}\label{sec-lemmas}


In this section we explore $\Delta$-encodings, which act on the extended phase space $(V',\omega')$. $\Delta$-encodings and their properties (particularly the existence of symplectic $\Delta$-encodings) are used both in the proof of the group structure theorem for $\pcl_{d,n}$ and in the proof of the bijection with $V^*\times\spgroup(V,\omega)$ (condensed encodings).

\begin{definition}
A \textit{$\Delta$-encoding} of a projective Clifford $\gamma\in\pcl_{d,n}$ is a pair $(\lambda,\phi)$ of functions $\lambda:V'\to R$ and $\phi:V'\to V'$ such that $\gamma(\Delta_v)=\zeta^{\lambda v}\Delta_{\phi v}$ for all $v\in V'$. We call a $\Delta$-encoding \textit{linear} if $\phi$ is ($R'$-)linear and \textit{symplectic} if $\phi\in\spgroup(V',\omega')$, or equivalently, if $\phi$ is linear and $\lambda$ is a homomorphism. We call two $\Delta$-encodings $(\lambda_1,\phi_1)$ and $(\lambda_2,\phi_2)$ of the same operation $\gamma$ \textit{equivalent} and write $(\lambda_1,\phi_1)\sim(\lambda_2,\phi_2)$.
\end{definition}

To better understand the properties of $\Delta$-encodings, it is helpful to consider a generalization. Beyond projective Cliffords, we may consider center-respecting functions $\gamma:\pauligroup_{d,n}\to\pauligroup_{d,n}$, by which we mean $\gamma(\zeta^t\Delta_v)=\zeta^t\gamma(\Delta_v)$ for all $t\in\frac12R'$ and $v\in V'$. Such functions are equivalently determined by right-definite relations on $\pauligroup_{d,n}$ of the form $\Delta_v\mapsto\zeta^{\lambda v}\Delta_{\phi v}$.\footnote{In the even case, this is not left-total since it is only defined on the proper subset $\{\Delta_v\,|\,v\in V'\}\subset\pauligroup_{d,n}$. For basic properties of relations, see \cref{appendix}.} Even more generally, we may consider center-respecting left-total relations $\gamma$ on $\pauligroup_{d,n}$ of the form $\Delta_v\mapsto\zeta^{\lambda v}\Delta_{\phi v}$ (by abuse of the $\mapsto$ notation) that aren't necessarily right-definite, which is possible on account of the non-uniqueness of the expressions $\Delta_v$ for $v\in V'$. 

\begin{definition}
Any pair $(\lambda,\phi)$ of functions $\lambda:V'\to\frac12R'$ and $\phi:V'\to V'$, which we call a \textit{pre-$\Delta$-encoding}, defines a center-respecting left-total relation on $\pauligroup_{d,n}$ by $\Delta_v\mapsto\zeta^{\lambda v}\Delta_{\phi v}$.
\end{definition}

The most severe case, where the relation $\Delta_v\mapsto\zeta^{\lambda v}\Delta_{\phi v}$ is not right-definite, is not of direct interest to us, but it arises hypothetically in order to obtain a characterization of $\Delta$-encodings (\cref{cor-neccsuff}). For now, we are mostly concerned with the pre-$\Delta$-encodings that determine center-respecting functions, as well as $\Delta$-encodings, (which determine projective Clifford operations). Naturally, we extend the definitions of \textit{linear} and \textit{symplectic} to pre-$\Delta$-encodings, and again write $(\lambda_1,\phi_1)\sim(\lambda_2,\phi_2)$ for two pre-$\Delta$-encodings defining the same relation. We caution that for linear pre-$\Delta$-encodings, the condition that $\lambda$ be a homomorphism is not equivalent to the condition that $\phi\in\spgroup(V',\omega')$. This is only true for $\Delta$-encodings.

Thus, there are two axes of variants of the notion of \textit{$\Delta$-encoding}, but most of them are introduced as temporary support towards our goal of understanding how to encode projective Clifford operations, and this ultimately all leads to the condensed encodings of \cref{sec-recipe}. The types of the functions involved in these variants are shown in \cref{fig:Encodings}. 

In this section our goal is to show that symplectic $\Delta$-encodings exist by first proving the existence of unstructured $\Delta$-encodings and then refining them, as shown in \cref{fig:linenc-sketch}. Along the way, we generalize our results to pre-$\Delta$-encodings so that we may later use \cref{thm-rdenough} to avoid extra checks when describing projective Cliffords using $\Delta$-encodings.



\begin{figure*}
    \centering
\[\begin{tikzcd}[column sep=small]
	{\text{projective Clifford}} \\
	\\
	{\text{(pre-)}\Delta\text{-encoding}} \\
	\\
	{\text{linear (pre-)$\Delta$-encoding}} & {\phi~\text{linear}} \\
	\\
	{\text{symplectic (pre-)$\Delta$-encoding}} & \begin{array}{c} \phi \in \spgroup(V',\omega') \\ \lambda~\text{linear iff proper encoding (\cref{thm-rdenough})}  \end{array}
	\arrow["{\text{existence (\cref{lem-delta})}}"', from=1-1, to=3-1]
	\arrow["{\text{linearization (\cref{lem-prelin})}}"', from=3-1, to=5-1]
	\arrow["{\text{symplectic lift (\cref{lem-linenc})}}"', from=5-1, to=7-1]
\end{tikzcd}\]
    \caption{Proof sketch of \cref{thm-linenc}---every projective Clifford has a symplectic $\Delta$-encoding}
    \label{fig:linenc-sketch}
\end{figure*}

Even though an arbitrary element of $\pauligroup_{d,n}$ has the form $\zeta^t\Delta_v$ for $t\in\frac12R'$ and $v\in V'$, the following lemma shows that there is no harm in restricting the codomain of $\lambda$ to the subgroup $R\leq\frac12R'$, but only when considering $\Delta$-encodings; when considering pre-$\Delta$-encodings, $\lambda$ is valued in $\frac12R'$.
\begin{lemma}\label{lem-lambdacodomain}
    Let $\gamma\in\pcl_{d,n}$ and suppose $\lambda:V'\to\frac12R'$ and $\phi:V'\to V'$ are functions satisfying $\gamma(\Delta_v)=\zeta^{\lambda v}\Delta_{\phi v}$ for all $v\in V'$. Then the image of $\lambda$ is contained in the subgroup $R\leq\frac12R'$.
\end{lemma}

In contrast, for a $D$-encoding $(r,\psi)$ (eg. of the $S$ gate), $r:V\to\frac12R'$ may take half-element values. 

\begin{proof}
    Let $v\in V'$. Due to the property $\Delta_v^t=\Delta_{tv}$ and since $\gamma$ is a homomorphism, we have:
    \begin{align*}
        I=\gamma(I)=\gamma(\Delta_v^d)=(\gamma(\Delta_v))^d=\zeta^{d\lambda v}\Delta_{d\phi v}=\zeta^{d\lambda v+\frac d2\langle\overline{\phi v}\rangle_0}I=\zeta^{d\lambda v}I
    \end{align*}

    which implies $d\lambda v=0$, and thus $\lambda v\in R\leq\frac12R'$.
\end{proof}

\begin{lemma}\label{lem-delta}
    Let $\gamma:\pauligroup_{d,n}\to\pauligroup_{d,n}$ be a center-respecting function and let $\phi:V'\to V'$ be a function such that for any $v\in V'$ there exists $t\in\frac12R'$ such that $\gamma(\Delta_v)=\zeta^t\Delta_{\phi v}$. Then there exists a unique function $\lambda:V'\to\frac12R'$ such that $(\lambda,\phi)$ is a pre-$\Delta$-encoding of $\gamma$.
\end{lemma}

Every center-respecting function $\gamma:\pauligroup_{d,n}\to\pauligroup_{d,n}$ has such a $\phi$, and hence a pre-$\Delta$-encoding. Likewise, each $\gamma\in\pcl_{d,n}$ has a $\Delta$-encoding. Intuitively, this lemma says the ``$t$''s in $\gamma(\Delta_v)=\zeta^t\Delta_{\phi v}$ can be conglomerated into a function $\lambda$.

\begin{proof}
    Existence is trivial. For uniqueness, simply note that for each $v\in V'$, we may extract the phase exponent $\lambda v$ from the quantity $\zeta^{\lambda v}\Delta_{\phi v}$, thereby determining $\lambda$.
\end{proof}

\begin{lemma}\label{lem-lambda0extends}
    For any functions $\lambda_0:V\to\frac12R'$ and $\phi:V'\to V'$, there exists a unique function $\lambda:V'\to\frac12R'$ extending $\lambda_0$ (with respect to the set inclusion $V\subseteq V'$) such that the relation $\Delta_v\mapsto\zeta^{\lambda v}\Delta_{\phi v}$ is right-definite, or equivalently, is the pre-$\Delta$-encoding of some center-respecting function $\gamma$.
\end{lemma}

\begin{proof}
    First let $\lambda$ be any function extending $\lambda_0$, and consider the center-respecting relation $\gamma$ on $\pauligroup_{d,n}$ given by $\Delta_v\mapsto\zeta^{\lambda v}\Delta_{\phi v}$. Due to the center-respecting property, every element of $\pauligroup_{d,n}$ has one or two images under $\gamma$, differing by sign if there are two. If $v\in V\subseteq V'$ and $\epsilon\in\mathbb Z_{d'/d}^{2n}$, then the two possible images of the quantity
    \begin{align*}
        \Delta_{v+d\epsilon} = (-1)^{\langle\epsilon\rangle_v}\Delta_v
    \end{align*}

    are:
    \begin{align*}
        \pm\zeta^{\lambda(v+d\epsilon)}\Delta_{\phi(v+d\epsilon)} = \pm(-1)^{\langle\epsilon\rangle_v}\zeta^{\lambda_0v}\Delta_{\phi v}
    \end{align*}

    We may therefore define $\delta$ by $d\delta=\phi v-\phi(v+d\epsilon)$, which does not depend on $\lambda$. Now consider the function $\lambda$ defined as follows. For each $v\in V\subseteq V'$ and $\epsilon\in\mathbb Z_{d'/d}$, set:
    \begin{align}\label{eqn-extendlambda}
    \lambda(v+d\epsilon)=\lambda_0v+\frac d2(\langle\epsilon\rangle_v+\langle\delta\rangle_{\phi v})&&\textrm{where}&&d\delta = \phi v - \phi(v+d\epsilon)
    \end{align}

    Then $\lambda$ is defined on all of $V'$ and extends $\lambda_0$. Moreover, for each $v+d\epsilon\in V'$ we have:
    \begin{align}\label{secondlemeqn}
        \zeta^{\lambda(v+d\epsilon)}\Delta_{\phi(v+d\epsilon)} =(-1)^{\langle\epsilon\rangle_v}\zeta^{\lambda v}\Delta_{\phi v}
    \end{align}

    Finally, $\gamma$ will be a function if and only if it is right-definite. For this, suppose $\Delta_{v_1}=\Delta_{v_2}$ for some $v_1,v_2\in V'$, and let $v\in V\subseteq V'$ be (the inclusion in $V'$) of their reduction mod $d$. Then letting $\epsilon_i\in\mathbb Z_{d'/d}$ satisfy $v_i=v+d\epsilon_i$ (for $i=1,2$), \cref{secondlemeqn} implies:
    \begin{align*}
        \zeta^{\lambda v_1}\Delta_{\phi v_1} = (-1)^{\langle\epsilon_1\rangle_{v}}\zeta^{\lambda v}\Delta_{\phi v} = (-1)^{\langle\epsilon_1\rangle_{v} + \langle\epsilon_2\rangle_{v}}\zeta^{\lambda v_2}\Delta_{\phi v_2} = \zeta^{\lambda v_2}\Delta_{\phi v_2}
    \end{align*}

    This establishes right definiteness. Uniqueness follows by \cref{lem-delta}.
\end{proof}

\begin{lemma}\label{lem-prelin}
    If $(\lambda_0,\phi_0)$ is a pre-$\Delta$-encoding and $\overline\phi_0:V'\to V$ is $R'$-linear, then there exists a canonical linear pre-$\Delta$-encoding $(\lambda,\phi)$ such that $(\lambda_0,\phi_0)\sim(\lambda,\phi)$
\end{lemma}

As we saw in \cref{sec-deltaenc}, if $(\lambda_0,\phi_0)$ is a $\Delta$-encoding, then $\overline\phi_0$ is always $R'$-linear.

\begin{proof}
    We call this process \textit{linearization}. Let $\phi:V'\to V'$ be the unique $R'$-linear map satisfying $\phi b=\phi_0b$ for each standard basis vector $b$ of $V'$. Since $\overline{\phi_0}$ is linear, $\overline{\phi v}=\overline{\phi_0v}$ for all $v\in V'$, allowing us to define $\lambda v=\lambda_0v + \frac d2\langle\delta\rangle_{\phi v,\phi_0v}$. Then by construction, $(\lambda,\phi)\sim(\lambda_0,\phi_0)$ and $\phi$ is linear.
\end{proof}


\begin{theorem}\label{thm-linlift}
    Let $\omega:V\otimes_RV\to R$ and $\omega':V'\otimes_{R'}V'\to R'$ be the canonical symplectic forms for the standard bases for $V$ and $V'$, respectively. Then each $\psi\in\spgroup(V,\omega)$ lifts to some $\phi\in\spgroup(V',\omega')$ (meaning $\psi=\overline{\overline\phi}$).
\end{theorem}

\begin{proof}
    See Theorem 1 on page 85 of \cite{newmansmart}.
\end{proof}

By use of condensed encodings, we will never need to \textit{compute} a symplectic lift $\phi\in\spgroup(V',\omega')$ of a given $\psi\in\spgroup(V,\omega)$, eg. in the computations of $\gamma(\Delta_v)$ or $\gamma^{-1}(\Delta_v)$; the existence of these liftings is only used in proofs. Still, we note that the lifting can be computed by solving a system of $4n^2$ equations and variables over $\mathbb F_2$, (assuming $d$ is even, of course). To see why, start by lifting $\psi$ to an $R'$-linear map $\phi_0$, and let $\Phi_0$ and $\Omega'$ be the $R'$-matrices for $\phi_0$ and $\omega_0$ in the standard basis. Then:
\begin{align}\label{searchforasol}
    \Phi_0^T\Omega'\Phi_0=\Omega'+dA
\end{align}
for some $\mathbb F_2$-matrix $A$. Then solve for an $\mathbb F_2$-matrix $E=(\begin{smallmatrix}E_1&E_2\\E_3&E_4\end{smallmatrix})$ such that:
\begin{align*}
    A = \begin{pmatrix}
            F_1^TE_3+F_3^TE_1+E_3^TF_1+E_1^TF_3 & F_1^TE_4+F_3^TE_2-E_3^TF_2+E_1^TF_4\\
            F_2^TE_3+F_4^TE_1+E_4^TF_1+E_2^TF_3 & F_2^TE_3 + F_4^TE_2 + E_4^TF_2 + E_2^TF_4
        \end{pmatrix}
\end{align*}
where $\Phi_0=(\begin{smallmatrix}F_1&F_2\\F_3&F_4\end{smallmatrix})$. This is a linear system in the entries of $E$. By design, $A=E^T\Omega'\Phi_0+\Phi_0^T\Omega'E$, so setting $\Phi=\Phi_0+dE$, the following equation shows that a solution to the linear system yields a symplectomorphism:
\begin{align*}
        \Phi^T\Omega'\Phi = \Phi_0^T\Omega'\Phi_0+d(E^T\Omega'\Phi_0+\Phi_0^T\Omega'E) = \Phi_0^T\Omega'\Phi_0 + dA = \Omega'
    \end{align*}
Since the reduction mod $d$ of $\Phi$ is the $R$-matrix $\Psi$ for $\psi$, \cref{thm-linlift} guarantees a solution.

\begin{lemma}\label{lem-linenc}
    Let $(\lambda_0,\phi_0)$ be a linear pre-$\Delta$-encoding such that $\psi=\overline{\overline\phi}_0\in\spgroup(V,\omega)$. Then $(\lambda_0,\phi_0)$ is equivalent to a symplectic pre-$\Delta$-encoding $(\lambda,\phi)$.
\end{lemma}

As seen in \cref{sec-deltaenc}, if $(\lambda_0,\phi_0)$ is a $\Delta$-encoding, then it is always the case that $\psi\in\spgroup(V,\omega)$.

\begin{proof}
    By \cref{thm-linlift}, there is a linear map $\epsilon:\mathbb Z_{d'/d}^{2n}\to\mathbb Z_{d'/d}^{2n}$ such that $\phi=\phi_0+d\epsilon\in\spgroup(V',\omega')$. Then set:
    \begin{align*}
        \lambda v = \lambda_0v + \frac d2\langle\epsilon v\rangle_{\phi_0v,\phi v}
    \end{align*}

    Then $(\lambda_0,\phi_0)\sim(\lambda,\phi)$ and $\phi\in\spgroup(V',\omega')$.
\end{proof}

\begin{theorem}\label{thm-linenc}
    Each $\gamma\in\pcl_{d,n}$ has a symplectic $\Delta$-encoding.
\end{theorem}

\begin{proof}
    This follows directly from 
    \cref{lem-delta,lem-prelin,lem-linenc} applied to the case of $\Delta$-encodings.
\end{proof}

\section{Necessary and sufficient conditions for \(\Delta\)-encodings}\label{sec-neccsuff}

In \cref{sec-lemmas} we showed that one may obtain a symplectic encoding $(\lambda,\phi)$ for a given $\gamma\in\pcl_{d,n}$.
In this section we prove the converse: starting with a pair $(\lambda,\phi)$ where $\lambda:V'\to R$ is a homomorphism and $\phi\in\spgroup(V',\omega')$, \cref{thm-rdenough} establishes that $(\lambda,\phi)$ encodes some $\gamma\in\pcl_{d,n}$. We use this later to prove the condensed encoding bijection. \cref{thm-rdenough} was proved for single qu$d$it in \citet{App05}, but our proof is entirely different.

By \cref{lem-lambdacodomain}, if $(\lambda,\phi)$ is the $\Delta$-encoding of some projective Clifford, then the image of $\lambda$ lies in $R\leq\frac12R'$. We must of course assume that $\lambda$ is $R$-valued to prove the converse result, but it is helpful to see how things would break if $\lambda$ were allowed to take half-element values. For example, in the single qubit case, take $\lambda:V'\to\frac12R'$ to be the homomorphism extending $\lambda(0,1)=\frac12$ and $\lambda(1,0)=0$, and let $\phi=\mathrm{id}_{V'}$. Then  $(\lambda,\phi)$ is not the $\Delta$-encoding of any projective Clifford, for if it were, it would follow that:
\begin{align*}
    +I=\gamma(\Delta_{01}^2)=\left(\gamma(\Delta_{01})\right)^2=-I
\end{align*}

The following lemma offers a more specific explanation: the problem is that the relation $\Delta_v\mapsto\zeta^{\lambda v}\Delta_{\phi v}$ (for $v\in V'$) is not right-definite; if it were, it would automatically be a homomorphism.

\begin{lemma}\label{lem-rdfail}
    Let $\gamma:\pauligroup_{d,n}\to\pauligroup_{d,n}$ be a center-respecting function satisfying $\gamma(\Delta_v)=\zeta^{\lambda v}\Delta_{\phi v}$ for some homomorphism $\lambda:V'\to\frac12R'$ and $\phi\in\spgroup(V',\omega')$. Then $\gamma\in\pcl_{d,n}$.
\end{lemma}

\begin{proof}
    Since $\gamma$ respects the center, the following calculation shows $\gamma$ is an endomorphism:
    \begin{align*}
        \gamma(\Delta_{v_1})\gamma(\Delta_{v_2})=\zeta^{\lambda v_1+\lambda v_2+\frac12\omega'(\phi v_1,\phi v_2)}\Delta_{\phi v_1+\phi v_2}=\zeta^{\lambda(v_1+v_2)+\frac12\omega'(v_1,v_2)}\Delta_{\phi(v_1+v_2)}=\gamma(\Delta_{v_1}\Delta_{v_2})
    \end{align*}

    Since $\phi$ is surjective and elements of $\pauligroup_{d,n}$ have the form $\zeta^t\Delta_v$, it follows that $\gamma:\pauligroup_{d,n}\to\pauligroup_{d,n}$ is surjective. Finally, since $\pauligroup_{d,n}$ is finite, $\gamma$ is a bijection, hence a projective Clifford.
\end{proof}

\begin{lemma}\label{lem-rdenoughmainlem}
    Let $\lambda:V'\to R$ be a homomorphism and let $\phi\in\spgroup(V',\omega')$. Then the relation $\gamma$ given by $\Delta_v\mapsto\zeta^{\lambda v}\Delta_{\phi v}$ is right definite, so extends to a center-respecting function $\pauligroup_{d,n}\to\pauligroup_{d,n}$.
\end{lemma}

\begin{proof}
    Consider the converse relation of the extended relation:
    \begin{align*}
        \gamma^{-1} = \left\{(\zeta^{t+\lambda v}\Delta_{\phi v},\zeta^t\Delta_v)\,\left\lvert\,t\in\frac12R'\textrm{ and }v\in V'\right.\right\}
    \end{align*}
    Then $\gamma^{-1}$ is a homomorphic relation: ($i$) $(I,I)\in\gamma^{-1}$; ($ii$) the inverse of $(\zeta^{t+\lambda v}\Delta_{\phi v},\zeta^t\Delta_v)$ is $(\zeta^{-t-\lambda v}\Delta_{-\phi v},\zeta^{-t}\Delta_{-v})\in\gamma^{-1}$; and ($iii$) if $(\zeta^{t_1+\lambda v_1}\Delta_{\phi v_1},\zeta^{t_1}\Delta_{v_1})\in\gamma^{-1}$ and $(\zeta^{t_2+\lambda v_2}\Delta_{\phi v_2}, \allowbreak \zeta^{t_2}\Delta_{v_2})\in\gamma^{-1}$, then their product is:
    \begin{align*}
        \,(\zeta^{t_3+\lambda v_3}\Delta_{\phi v_3},\zeta^{t_3}\Delta_{v_3})&=(\zeta^{t_1+t_2+\frac12\omega'(\phi v_1,\phi v_2)+\lambda v_1+\lambda v_2}\Delta_{\phi v_1+\phi v_2},\zeta^{t_1+t_2+\frac12\omega'(v_1,v_2)}\Delta_{v_1+v_2})\\
        &=\,(\zeta^{t_1+t_2+\frac12\omega'(v_1,v_2)+\lambda(v_1+v_2)}\Delta_{\phi(v_1+v_2)},\zeta^{t_1+t_2+\frac12\omega'(v_1,v_2)}\Delta_{v_1+v_2})\in\gamma^{-1}
    \end{align*}

    Now, to see that $\gamma$ is right-definite, we show that $\gamma^{-1}$ is left-definite using \cref{lem-injlem}. Let $x\in\mathrm{ker\,}\gamma^{-1}$. Then $x=\zeta^{t+\lambda v}\Delta_{\phi v}$ for some $t\in\frac12R'$ and $v\in V'$ such that $\zeta^t\Delta_v=I$. This implies that $v=d\epsilon$ for some $\epsilon\in\mathbb Z_{d'/d}^{2n}$ and $t=\frac d2\langle\epsilon\rangle_0=0$. Then $x=\zeta^{\lambda d\epsilon}\Delta_{\phi d\epsilon}=\zeta^{d\lambda\epsilon}\Delta_{d\phi\epsilon}$ since $\lambda$ and $\phi$ are homomorphisms. Since $\lambda$ is valued in $R\leq\frac12R'$, $\lambda d\epsilon=0$, so $x=\Delta_{d\phi\epsilon}=\zeta^{\frac d2\langle\phi\epsilon\rangle_0}I=I$. So $\mathrm{ker\,}\gamma^{-1}$ is trivial.
\end{proof}

\begin{theorem}\label{thm-rdenough}
    Let $\lambda:V'\to R$ be a homomorphism and let $\phi\in\spgroup(V',\omega')$. Then the relation $\Delta_v\mapsto\zeta^{\lambda v}\Delta_{\phi v}$ is right-definite and extends to a projective Clifford $\gamma\in\pcl_{d,n}$.
\end{theorem}

\begin{proof}
This follows from \cref{lem-rdfail,lem-rdenoughmainlem}. Alternately, first use \cref{lem-rdenoughmainlem} to establish that $\gamma$ is a function. Next (as in the proof of \cref{lem-rdfail}), $\gamma$ is bijective. Then the converse relation $\gamma^{-1}$ is just the inverse function, and since it is a homomorphic relation, it is a homomorphism. Finally $\gamma$ is a homomorphism since its inverse is.
\end{proof}

We now arrive at a necessary and sufficient condition for $(\lambda,\phi)$ to determine a projective Clifford.

\begin{corollary}\label{cor-neccsuff}
    Let $\lambda:V'\to R$ and $\phi:V'\to V'$ be a pair of functions whose relation $\Delta_v\mapsto\zeta^{\lambda v}\Delta_{\phi v}$ is right-definite, or equivalently, whose relation extends to a center-respecting function $\gamma:\pauligroup_{d,n}\to\pauligroup_{d,n}$. Also let $\lambda_0:V\to R$ be the restriction of $\lambda$ to $V\subseteq V'$. Then $\gamma$ is projective Clifford if and only if $\lambda_0$ is a homomorphism, $\overline\phi:V'\to V$ is $R'$-linear, and furthermore, the map $\psi = \overline{\overline{\phi}} :V\to V$ obtained from $\phi$ by reducing the output and input of $\phi$ mod $d$, (which is well-defined since $\overline\phi$ is linear) is a symplectomorphism, $\psi\in\spgroup(V,\omega)$.
\end{corollary}

\begin{proof}
    If either of $\lambda_0$ or $\overline\phi$ is not linear or $\psi$ is not a symplectomorphism, then by preliminary observations in \cref{sec-deltaenc}, $\gamma\notin\pcl_{d,n}$, so assume they are linear and $\psi\in\spgroup(V,\omega)$. Then using the full strength of \cref{lem-prelin,lem-linenc}, we may assume without loss of generality that $\phi\in\spgroup(V,\omega)$. Then since $\lambda_0$ is a homomorphism extended by a function $\lambda$, by the ``uniqueness'' part of \cref{lem-lambda0extends}, $\lambda$ is a homomorphism. So by \cref{thm-rdenough}, $\gamma\in\pcl_{d,n}$.
\end{proof}


Lastly, if $\phi$ is linear, then right-definiteness of the relation $\Delta_v\mapsto\zeta^{\lambda v}\Delta_{\phi v}$ depends on $\phi$ and not $\lambda$. Indeed, if $\phi$ is linear, then the relation is right-definite if and only if
\begin{align}\label{eqn-rdcheck}
    \langle\epsilon\rangle_v=\langle\phi\epsilon\rangle_{\phi v}
\end{align}
for all $v\in V'$ and $\epsilon\in\mathbb Z_{d'/d}^{2n}$.

\section{Structure of the projective Clifford group and cohomology}\label{sec-structurethm}

We now present the Structure Theorem for the projective Clifford group (\cref{thm-structure}), which provides a group isomoprhism $\spgroup(V,\omega)\cong\pcl_{d,n}/V^*$. The $d=2$ case was shown previously by \citet{bolt1961clifford}.


Given a $\Delta$-encoding $(\lambda,\phi)$ of some $\gamma\in\pcl_{d,n}$, let us write $\psi=\overline{\overline\phi}\in\spgroup(V,\omega)$, as in \cref{sec-deltaenc}.

\begin{theorem}\label{thm-structure}
    The function $\gamma\mapsto\psi$ is a well-defined surjective group homomorphism $\pi:\pcl_{d,n}\to\spgroup(V,\omega)$ with kernel $i:V^*\to\pcl_{d,n}$ identifying $\mu\in V^*$ with $(\lambda,\mathrm{id}_{V'})$ for the unique $\lambda\in(V')^*$ such that $\lambda v=\mu\overline v$ for all $v\in V'$.
    \begin{equation*}
        \begin{tikzcd}
        1\ar[r]&V^*\ar[r,"i"]&\pcl_{d,n}\ar[r,"{\pi}"]&\spgroup(V,\omega)\ar[r]&1
        \end{tikzcd}
    \end{equation*}
\end{theorem}

\begin{proof}
    If $(\lambda_1,\phi_1)$ and $(\lambda_2,\phi_2)$ are $\Delta$-encodings of $\gamma\in\pcl_{d,n}$, then for each $v\in V'$ there exists $\epsilon\in\mathbb Z_{d'/d}^{2n}$ such that $\phi_1v=\phi_2v+d\epsilon$, so $\gamma\mapsto\psi$ is well-defined. Now if $\gamma$ and $\gamma'$ have $\Delta$-encodings $(\lambda,\phi)$ and $(\lambda',\phi')$ respectively, then their composition acts by:
    \begin{align}
        \gamma\gamma'(\Delta_v)=\zeta^{\lambda'v+\lambda\phi'v}\Delta_{\phi\phi'v}
    \end{align}

    Reading subscripts mod $d$, we see that $\gamma\mapsto\psi$ is multiplicative. 
    It follows that $\gamma\mapsto\psi$ is a group homomorphism $\pi:\pcl_{d,n}\to\spgroup(V,\omega)$. To see that $\pi$ is surjective, by \cref{thm-linlift}, lift an arbitrary $\psi\in\spgroup(V,\omega)$ to some $\phi\in\spgroup(V',\omega')$. By \cref{thm-rdenough}, the relation $\Delta_v\mapsto\Delta_{\phi v}$ is right-definite and extends to a projective Clifford, and since it has $\Delta$-encoding $(0,\phi)$, it follows that $\pi\gamma=\psi$.

    We now describe the kernel. Every $\gamma\in\mathrm{ker\,}\pi$ has a $\Delta$-encoding of the form $(\lambda,\mathrm{id}_{V'})$, since by \cref{lem-delta} it has a $\Delta$-encoding $(\lambda_0,\phi_0)$, and then $(\lambda_0,\phi_0)\sim(\lambda,\mathrm{id}_{V'})$ where $\lambda=\lambda_0+\frac d2\langle\epsilon\rangle_{\phi_0v,v}$. Since $\mathrm{id}_{V'}\in\spgroup(V',\omega')$, it follows that $\lambda$ is a homomorphism. It follows that there exists unique $\mu\in V^*$ such that $\lambda v=\mu\overline v$ for all $v\in V'$. Conversely, if $\lambda$ corresponds in this way to an arbitrary $\mu\in V^*$, then $(\lambda,\mathrm{id}_{V^*})$ is a $\Delta$-encoding of some $\gamma\in\pcl_{d,n}$. This can be seen directly or by \cref{thm-rdenough}.
\end{proof}

By Lagrange's Theorem, $|\pcl_{d,n}|=|V^*||\spgroup(V,\omega)|$. Recall that one proves Lagrange's Theorem by constructing a bijection between the subgroup in question and an arbitrary coset. In order to obtain a bijection between the sets $\pcl_{d,n}$ and $V^*\times\spgroup(V,\omega)$, one must choose a section $s:\spgroup(V,\omega)\to\pcl_{d,n}$ of the quotient map $\pi$. The standard basis on $V=\mathbb Z_d^{2n}$ provides a canonical choice. Namely, one lets the entries of $s(\psi)$ (in $R'$) be the inclusions of the matrix entries for $\psi$ (in $R$) via the inclusion function $\underline{\,\cdot\,}:R\to R'$, or in abbreviation, $s=\underline{\,\cdot\,}$. More generally, any (setwise) section $s:G\to B$ of a surjective group homomorphism $\pi:B\to G$ with kernel $A$ yields a bijection:
\begin{align}\label{sectioniso}
    B=\bigsqcup_{g\in G}As(g)&\to A\times C\\
    as(g)&\mapsto(a,g)
\end{align}

This leaves something to be desired, however, since we would like a formula for the action of a projective Clifford $\gamma\in\pcl_{d,n}$, encoded as a pair $(\mu,\psi)\in V^*\times\spgroup(V,\omega)$, on a Weyl operator $\Delta_v$. Group cohomology gets us there in an abstract sense, but more steps are required (shown in \cref{sec-recipe}) to carry out the calculation concretely. For reference, we briefly summarize the relevant cohomological technique, which is outlined in \citet{conrad2018group} and is applied to this example in \citet{raussendorf2023role}.

Suppose $A$ is a right $G$-module. This means $A$ is an abelian group equipped with a right action $\cdot:A\times G\to A$. In our application, this is just composition (pullback): $\psi\cdot\mu=\psi\mu$. Then in terms of \cref{sectioniso}, the group law has the form
\begin{align*}
    (a_1,g_1)(a_2,g_2) = (a_1+a_2\cdot g_1+c(g_1,g_2),g_1g_2)
\end{align*}

and one may verify that $c:G\times G\to A$ is a group 2-cocycle. A different section $s'$ results in a different 2-cocycle $c'$, but always of the same cohomology class $[c]$. Conversely, each 2-cocycle $c$ yields a group extension $A\to B\to G$ whose isomorphism class is uniquely determined by $[c]$. We calculate $c(\psi_1,\psi_2)$ in the case of $V^*\to\pcl_{d,n}\to\spgroup(V,\omega)$ in the case that $s=\underline{\,\cdot\,}$ in \cref{sec-recipe}.


\section{Condensed \(\Delta\)-encodings}\label{sec-recipe}

We give a constructive proof of existence and uniqueness of condensed encodings, and then derive the formulas for their evaluation, composition, and inversion. In this section and the next, we will frequently use the inclusion function $\underline{\,\cdot\,}:V\to V'$ or $R\to R'$, which unlike the inclusion $R\leq\frac12R'$ (which we do not indicate notationally with an underline or otherwise), is not a homomorphism when $d$ is even. We also make frequent use of the reduction mod $d$ homomorphism $\overline{\,\cdot\,}:V'\to V$ or $R'\to R$. The relevant property is that $\underline{\,\cdot\,}$ is a setwise section of $\overline{\,\cdot\,}$.

We use the term \textit{extended encoding} to mean \textit{$\Delta$-encoding} in this section since both the condensed and extended encoding involve the Weyl operators $\Delta_v$.

\begin{theorem}\label{thm-denseenc}
    There is a bijection $\pcl_{d,n}\to V^*\times\spgroup(V,\omega)$ such that for any symplectic extended encoding $(\lambda,\phi)$ of $\gamma\in\pcl_{d,n}$, the corresponding pair $(\mu,\psi)\in V^*\times\spgroup(V,\omega)$ uniquely satisfies:
    \begin{align}\label{eqn-bijfundamentals}
    \psi=\overline{\overline\phi} && \textrm{and} && \lambda\underline b-\mu b=\frac d2\langle\epsilon\rangle_{\underline{\psi b},\phi\underline b}
    \end{align}

    for all standard basis vectors $b\in V$.
\end{theorem}

\begin{proof}
    Given $\gamma\in\pcl_{d,n}$, we define $(\mu,\psi)$ as follows. By \cref{thm-linenc}, let $(\lambda,\phi)$ be a symplectic extended encoding of $\gamma$. Then for each $v\in V'$, we have:
    \begin{align*}
        \gamma(\Delta_v) = \zeta^{\lambda v}\Delta_{\phi v}=\zeta^{\lambda v+\frac d2\langle\epsilon\rangle_{\underline{\psi\overline v},\phi v}}\Delta_{\underline{\psi\overline v}}
    \end{align*}

    where $\psi=\overline{\overline\phi}\in\spgroup(V,\omega)$. Now define $\mu:V\to R$ by its action on the standard basis vectors $b\in V$:
    \begin{align*}
        \mu b:=\lambda\underline b
    \end{align*}

    By construction, $\mu$ is $R$-linear ($\mu\in V^*$), and we obtain the following formula for evaluating $\gamma(\Delta_{\underline b})$. Note that elements $\Delta_{\underline b}$ together with $\tau$ generate $\pauligroup_{d,n}$.
    \begin{align}\label{eqn-basiseval}
        \gamma(\Delta_{\underline b}):=\zeta^{\mu b}\Delta_{\underline{\psi b}}
    \end{align}

    We now argue in two (very) different ways that the function $\gamma\mapsto(\mu,\psi)$ just described is a bijection. First, injectivity is straight-forward. Since the $\Delta_{\underline b}$ together with $\tau$ generate $\pauligroup_{d,n}$, we have:
    \small\begin{align*}
        (\mu_1,\psi_1)=(\mu_2,\psi_2)
        \,\,\implies\,\,
        \gamma_1(\Delta_{\underline b})=\gamma_2(\Delta_{\underline b})\textrm{ for standard basis vectors }b\in V
        \,\,\implies\,\,
        \gamma_1=\gamma_2
    \end{align*}\normalsize
    The result follows from \cref{thm-structure} since $|\pcl_{d,n}|\geq|V^*||\spgroup(V,\omega)|$. 
    
    Alternately, we construct the inverse as follows. By \cref{thm-linlift}, given $(\mu,\psi)$, we may lift $\psi$ to $\phi\in\spgroup(V',\omega')$, at which point for standard basis vectors $b\in V$ we have:
    \begin{align*}
        \zeta^{\mu b}\Delta_{\underline{\psi b}}=\zeta^{\mu b+\frac d2\langle\epsilon\rangle_{\underline{\psi b},\phi\underline b}}\Delta_{\phi\underline b}
    \end{align*}

    Then define $\lambda:V'\to R$ by $R'$-linearly extending its action on the standard basis vectors:
    \begin{align*}
        \lambda\underline b=\mu b+\frac d2\langle\epsilon\rangle_{\underline{\psi b},\phi\underline b}
    \end{align*}

    By \cref{thm-rdenough}, $(\lambda,\phi)$ is a symplectic encoding of some (unique) $\gamma\in\pcl_{d,n}$, and \cref{eqn-basiseval} is satisfied. Since \cref{eqn-basiseval} uniquely determines both $\gamma$ and $(\mu,\psi)$, these constructions are mutually inverse.
\end{proof}



\begin{theorem}\label{thm-eval}
    Let $\gamma\leftrightarrow(\mu,\psi)$ correspond as in \cref{thm-denseenc}. Then for each $v\in V'$ we have:
    \begin{align}\label{eqn-evalformula}
        \gamma(\Delta_v)=\zeta^{\mu\overline v + \frac d2\kappa'v}\Delta_{\underline{\psi\overline v}}
    \end{align}

    where the function $\kappa':V'\to\mathbb Z_{d'/d}$ is defined by:
    \begin{align*}
        \frac d2\kappa'v=\frac12 z\cdot x+\frac12\sum_{i=1}^n\omega'(x_i\underline{\psi(e_i,0)},z_i\underline{\psi[0,e_i]})+\frac d2\sum_{i=1}^n\langle\epsilon\rangle_{x_i\underline{\psi(e_i,0)}+z_i\underline{\psi(0,e_i)}\,,\,\underline{\psi\overline v}}
    \end{align*}
\end{theorem}

Since $\langle\epsilon\rangle$ can be computed using \cref{eqn-computeeps}, \cref{thm-eval} gives an algorithm for computing $\gamma(\Delta_v)$. Sometimes it is also helpful to introduce the function $\kappa:V\to\mathbb Z_{d'/d}$ given by $\kappa v=\kappa'\underline v$.

\begin{proof}
    We obtain:
    \begin{align*}
        \gamma(\Delta_v)&=\tau^{z\cdot x}\prod\gamma(X_i)^{x_i}\gamma(Z_i)^{z_i}\\
        &=\tau^{z\cdot x}\prod(\zeta^{x_i\mu(e_i,0)}\Delta_{x_i\underline{\psi(e_i,0)}})(\zeta^{z_i\mu(0,e_i)}\Delta_{z_i\underline{\psi(0,e_i)}})\\
        &=\zeta^{\frac12z\cdot x+\mu v+\frac12\sum\omega'(x_i\underline{\psi(e_i,0)},z_i\underline{\psi(0,e_i)})}\prod\Delta_{x_i\underline{\psi(e_i,0)}+z_i\underline{\psi(0,e_i)}}
    \end{align*}

    On the other hand, for each $v\in V'$ there exists unique $k\in\mathbb Z_{d'/d}$ such that $\gamma(\Delta_v)=\zeta^{\mu\overline v+\frac d2k}\Delta_{\underline{\psi\overline v}}$. Since the subscripts to the $\Delta$ operator agrees with \cref{eqn-evalformula}, $\kappa'v=k$ is as in the theorem statement.
\end{proof}

\begin{corollary}\label{cor-kappaformula}
    Let $\gamma\leftrightarrow(\mu,\psi)$ as in \cref{thm-denseenc}, let $(\lambda,\phi)$ be any symplectic extended encoding of $\gamma$, and let $\kappa'$ be its evaluation phase correction function as in \cref{thm-eval}. Then for each $v\in V'$ we have:
\begin{align}\label{eqn-bijfundamentals2}
    \lambda v-\mu\overline v=\frac d2\langle\epsilon\rangle_{\underline{\psi\overline v},\phi v} + \frac d2\kappa'v
\end{align}
\end{corollary}

\cref{eqn-bijfundamentals2} generalizes \cref{eqn-bijfundamentals} by allowing for vectors beyond just standard basis vectors.

\begin{proof}
    Compare the phases of the left- and right-hand sides of the following equation:
    \begin{align*}
        \zeta^{\lambda v+\frac d2\langle\epsilon\rangle_{\underline{\psi\overline v},\phi v}}\Delta_{\underline{\psi\overline v}}=\zeta^{\lambda v}\Delta_{\phi v}=\zeta^{\mu\overline v+\frac d2\kappa'v}\Delta_{\underline{\psi\overline v}}
    \end{align*}
\end{proof}

\cref{eqn-bijfundamentals2} does not give a canonical symplectic extended encoding in terms of $\kappa'$, since for example, both symplectic encodings of $\mathrm{CNOT}$ seen in \cref{sec-examples} satisfy \cref{eqn-bijfundamentals2}. Conversely, however, not only can a symplectic encoding $(\lambda,\phi)$ be used to directly compute $\gamma(\Delta_v)=\zeta^{\lambda v}\Delta_{\phi v}$, but \cref{cor-kappaformula} provides an alternate way to compute $\kappa'$ in terms of $(\lambda,\phi)$. As noted before, a symplectic extended encoding can be obtained by solving a linear system. Still, \cref{thm-eval} provides a more direct approach to evaluation.

\begin{corollary}\label{cor-comp}
    Let $(\mu_1,\psi_1)$ and $(\mu_2,\psi_2)$ be condensed encodings of $\gamma_1$ and $\gamma_2$, respectively. Then the condensed encoding of the composition $\gamma_2\gamma_1$ is $(\mu_3,\psi_3)$, where $\psi_3=\psi_2\psi_1$ and $\mu_3$ is defined by $R$-linearly extending its action on standard basis vectors $b\in V$:
    \begin{align*}
        \mu_3b:=\mu_1b+\mu_2\psi_1b+\frac d2\kappa_2\psi_1b
    \end{align*}
\end{corollary}

We caution that $\mu_3v$ must be computed using linearity by expanding $v$ in the standard basis for $V$, rather than using the function $\tilde\mu$ below, which might not be linear.

\begin{proof}
    Here we use $\kappa$ instead of $\kappa'$. The evaluation formula reads $\gamma(\Delta_{\underline v})=\zeta^{\mu v+\frac d2\kappa v}\Delta_{\underline{\psi v}}$ for $v\in V$. Now extract the canonically determined functions $\tilde\mu_3$ and $\psi_3$ from $\zeta^{\tilde\mu_3v}\Delta_{\underline{\psi_3v}}:=\gamma_2\gamma_1(\Delta_{\underline v})$:
    \begin{align}\label{eqn-mutilde}
        \tilde\mu_3v=\mu_1v+\mu_2\psi_1v+\frac d2\kappa_1v+\frac d2\kappa_2\psi_1v
    \end{align}

    In particular, \cref{eqn-mutilde} holds for standard basis vectors $b\in V$, and together with $\psi_3=\psi_2\psi_1$, determines the condensed encoding as in \cref{eqn-basiseval}. Also, $\kappa_1 b=0$.
\end{proof}

\begin{corollary}
    The operation $((\mu_2,\psi_2),(\mu_1,\psi_1))\mapsto(\mu_3,\psi_3)$ makes the set $V^*\times\spgroup(V,\omega)$ into a group isomorphic to $\pcl_{d,n}$.
\end{corollary}

\begin{proof}
    This follows from \cref{cor-comp} and the bijection of \cref{thm-denseenc}.
\end{proof}

\begin{corollary}
    The identity element is $(0,\mathrm{id}_V)$ and the inverse of $(\mu,\psi)$ is $(\mu_{\mathrm{inv}},\psi^{-1})$, where:
    \begin{align*}
        \mu_{\mathrm{inv}}b = -\mu\psi^{-1}b+\frac d2\kappa\psi^{-1}b
    \end{align*}

    for all standard basis vectors $b\in V$.
\end{corollary}

\begin{proof}
    This follows from \cref{cor-comp}.
\end{proof}

\section{The condensed product \(\cprod\)}\label{sec-odot}

We consider the $\mathbb C$-bilinear function $\cprod:\pauligroup_{d,n}\times\pauligroup_{d,n}\to\pauligroup_{d,n}$ defined for each $u,v\in V$ by
\begin{align}\label{eqn-odotdef}
    \Delta_{\underline u}\cprod\Delta_{\underline v}=\tau^{-\underline{\omega(u,v)}}\Delta_{\underline u}\Delta_{\underline v}=(-1)^{s(u,v)}\Delta_{\underline u+\underline v}
\end{align}
where $s(u,v)=\frac1d(\omega'({\underline u},{\underline v})-\underline{\omega( u, v)})$. This equation extends to all $u,v\in V'$ as follows:
\begin{align*}
    \Delta_u\cprod\Delta_v = (-1)^{\langle\rho\rangle}\Delta_{\overline{\underline u}}\cprod\Delta_{\overline{\underline v}}=(-1)^{\langle\rho\rangle+s(\overline u,\overline v)}\Delta_{\underline{\overline u}+\underline{\overline v}}=(-1)^{s(\overline u,\overline v)}\Delta_{u+v}
\end{align*}
Observe that for each $t_1,t_2\in R$, the phase exponent of $\zeta^{t_1}\Delta_u\cprod\zeta^{t_2}\Delta_v$ is still in $R$:
\begin{align}\label{eqn-condensed-prod-encoding}
    \zeta^{t_1}\Delta_{\underline u}\cprod\zeta^{t_2}\Delta_{\underline v}=\zeta^{t_1+t_2+\frac d2s(u,v)+\frac d2\langle\rho\rangle_{\underline u+\underline v,\underline{u+v}}}\Delta_{\underline{u+v}}
\end{align}
This implies that the following subsets of $\pauligroup_{d,n}$ are closed under $\cprod$:
\begin{align*}
    \mathbf Q_{d,n}=\{\zeta^t\Delta_{\underline v}\,|\,t\in R\textrm{ and }v\in V\} && \mathbf\Sigma_{d,n}=\{\pm\Delta_v\,|\,v\in V'\}=\{\pm\Delta_{\underline v}\,|\,v\in V\}
\end{align*}
Notice that $\mathbf Q_{d,n}$ is in obvious bijection with the set $R\times V$. In applications, this encoding scheme for Paulis works well with the condensed encoding scheme for projective Cliffords.

Now let $R'_+=\{0,\cdots,d-1\}$ and $R'_-=\{d,\cdots,d'-1\}$ so $R'=R'_+\sqcup R'_-$. We write $\mathrm{sgn}(t)=0$ if $t\in R'_+$ and $\mathrm{sgn}(t)=1$ if $t\in R'_-$. For example, $s(u,v)=\mathrm{sgn}(\omega'(\underline u,\underline v))$. Then consider the function:
\begin{align*}
    f:\pauligroup_{d,n}\to\mathbf\Sigma_{d,n}&&f(\tau^t\Delta_v)=(-1)^{\mathrm{sgn}(t)}\Delta_v
\end{align*}
Note that $f$ is well-defined: $f(\tau^{t+\frac d2\langle\epsilon\rangle_v}\Delta_v)=(-1)^{\mathrm{sgn}(t)+\langle\epsilon\rangle_v}\Delta_v$. Though $f$ does not act homomorphically on the center (for a single qubit, $-I=f(-I)\neq f(iI)\cprod f(iI)=I\cprod I=I$), it satisfies:
\begin{align*}
    f(\Delta_u)\cprod f(\Delta_v)=(-1)^{s(\overline u,\overline v)}\Delta_{u+v}=f(\tau^{\omega'(u,v)}\Delta_{u+v})=f(\Delta_u\Delta_v)
\end{align*}
This prompts us to define the following equivalence relation $\sim$ on $\pauligroup_{d,n}$ and give the quotient set $\pauligroup_{d,n}/\sim$ the following operation, in order to obtain an isomorphism $\tilde f:\pauligroup_{d,n}/\sim\,\to\mathbf\Sigma_{d,n}$:
\begin{align*}
    \tau^s\Delta_v\sim\tau^t\Delta_{v+d\epsilon}
    \textrm{  iff 
 }
 \mathrm{sgn}(s)+\mathrm{sgn}(t)=\langle\epsilon\rangle_v 
 && 
 [\tau^s\Delta_u]\cprod[\tau^t\Delta_v]=[\tau^{r} (-1)^{s(\overline u,\overline v)}\Delta_{u+v}]
\end{align*}
where $r$ satisfies $\mathrm{sgn}(s)+\mathrm{sgn}(t)=\mathrm{sgn}(r)$. The intuition of the condensed product is therefore to only remember the ``sign'' of the phase exponent. For even $d>2$, $\star$ is non-associative:
\begin{align*}
    (X\cprod Y)\cprod Z=-Z\cprod Z=-Z^2\neq-X^2=-X\cprod X=X\cprod(Y\cprod Z)
\end{align*}

\begin{theorem}\label{thm-odotprop}
    Let $\gamma\in\pcl_{d,n}$ and $u,v\in V'$. Then:
    \begin{align*}
        \gamma(\Delta_u\cprod\Delta_v) = \gamma(\Delta_u)\cprod\gamma(\Delta_v)
    \end{align*}
\end{theorem}
\begin{proof}
    It suffices to just show the result holds for $\underline{v_1}$ and $\underline{v_2}$ with $v_1,v_2 \in V$. Let $(\mu,\psi)$ be the condensed encoding for $\gamma$ and let $\kappa$ be its correction function. Then,
    \begin{align*}
        \gamma(\Delta_{\underline{u}} \cprod \Delta_{\underline{v}})&= \zeta^{-\frac12\underline{\omega(u,v)}} \gamma(\Delta_{\underline{u}}) \gamma(\Delta_{\underline{v}})\\
        &=\zeta^{\mu u+\mu v+\frac d2(\kappa u+\kappa v)-\frac12\underline{\omega(u,v)}}\Delta_{\underline{\psi u}}\Delta_{\underline{\psi v}}\\
        &=\zeta^{\mu u+\mu v+\frac d2(\kappa u+\kappa v)-\frac12\underline{\omega(u,v)}+\frac12\underline{\omega(\psi u,\psi v)}}\Delta_{\underline{\psi u}}\cprod\Delta_{\underline{\psi v}}\\
        &=\zeta^{\mu u+\mu v+\frac d2(\kappa u+\kappa v)}\Delta_{\underline{\psi u}}\cprod\Delta_{\underline{\psi v}}\\
        &=\gamma(\Delta_{\underline u})\cprod\gamma(\Delta_{\underline v})
    \end{align*}
\end{proof}

\begin{corollary}
    For all $u,v\in V$ and $\psi\in\spgroup(V,\omega)$, we have $s(u,v)=s(\psi u,\psi v)$.
\end{corollary}
In other words, symplectomorphisms $\psi\in\spgroup(V,\omega)$ preserve the sign of $s(u,v)$.
\begin{proof}
    Pick some $\mu\in V^*$ (eg. $\mu=0$), and let $(\lambda,\phi)$ be a symplectic $\Delta$-encoding of the $\gamma\in\pcl_{d,n}$ whose condensed encoding is $(\mu,\psi)$. Then:
    \begin{align*}
        \gamma(\Delta_u\cprod\Delta_v) = (-1)^{s(\overline u,\overline v)+s(\overline{\phi u},\overline{\phi v})}\gamma(\Delta_u)\cprod\gamma(\Delta_v)
    \end{align*}
\end{proof}

\section{Examples of \(\Delta\)-encodings and condensed encodings}\label{sec-examples}

\begin{example}\label{example-cnot}
    Let $d=2$ and let $\gamma$ be conjugation by $\mathrm{CNOT}$. By considering the Pauli tableau ($X\otimes I\mapsto X\otimes X$, $I\otimes X\mapsto I\otimes X$, $Z\otimes I\mapsto Z\otimes I$, $I\otimes Z\mapsto Z\otimes Z$), the obvious way to get a $\Delta$-encoding (\textit{cf.} \cref{lem-prelin}) is to start with the following $R'$-linear map (written as a matrix in the standard basis):
    \begin{align*}
        \phi_0 = \begin{pmatrix}
            1\\1&1\\&&1&1\\&&&1
        \end{pmatrix}
    \end{align*}
    By \cref{lem-delta}, there exists unique $\lambda_0$ such that $(\lambda_0,\phi_0)$ is an encoding of $\gamma$, but it is not symplectic, since by considering the action on basis vectors, we would have $\lambda_0=0$, but $\lambda_0(e_1+e_4)\neq\lambda_0 e_1+\lambda_0 e_4$. 
    The following matrix calculation shows directly that $\phi_0\not\in\spgroup(V',\omega')$:
    \begin{align*}
        [\phi_0^*(\omega')-\omega'] = [\phi_0]^T[\omega'][\phi_0] - [\omega'] = dA = 2\begin{pmatrix}
            0&0&0&1\\0&0&0&0\\0&0&0&0\\1&0&0&0
        \end{pmatrix}
    \end{align*}
    By searching for solutions to \cref{searchforasol}, we obain 
    two symplectic encodings for CNOT (which both use $\lambda=0$) are given by:
    \begin{align}\label{example-cnot-twophis}
        \phi = \begin{pmatrix}
            1\\3&1\\&&1&1\\&&&1
        \end{pmatrix} && \textrm{or} && \phi=\begin{pmatrix}
            1\\1&1\\&&1&3\\&&&1
        \end{pmatrix}
    \end{align}
    The condensed encoding for CNOT is $(0,\psi)$ where $\psi=\overline{\overline\phi}$.
\end{example}

\begin{example}
    Either of $\phi_1=(\begin{smallmatrix}1&0\\1&1\end{smallmatrix})\in\spgroup(V',\omega')$ or $\phi_2=(\begin{smallmatrix}1&0\\3&1\end{smallmatrix})\in\spgroup(V',\omega')$ can be used in symplectic $\Delta$-encodings of (conjugation by) either $S$ or $S^{-1}$, each with the appropriate $\lambda$. The condensed encodings both $S$ and $S^{-1}$ use $\psi=(\begin{smallmatrix}1&0\\1&1\end{smallmatrix})\in\spgroup(V,\omega)$, and the respective phase functions are $\mu_S=0$ and $\mu_{S^{-1}}(\begin{smallmatrix}1\\0\end{smallmatrix})=1$ and $\mu_{S^{-1}}(\begin{smallmatrix}0\\1\end{smallmatrix})=0$. The associated ``correction functions'' $\kappa_S$ and $\kappa_{S^{-1}}$ are nonlinear. For example, $\kappa_S(\begin{smallmatrix}1\\1\end{smallmatrix})=1$ (since $SYS^{-1}=-X$) and $\kappa_Sv=0$ for $v\neq(\begin{smallmatrix}1\\1\end{smallmatrix})$.
\end{example}


\begin{example}\label{example-sgate}
    The quantum Fourier transform acts on a single \qudit by $X\mapsto Z$ and $Z\mapsto X^{-1}$. A symplectic encoding is $\lambda=0$ and $\phi=(\begin{smallmatrix}
        &-1\\1
    \end{smallmatrix})\in\spgroup(V',\omega')$. For $d=2$ this is the Hadamard gate.
    The condensed encoding is just $\mu=0$ and $\psi=(\begin{smallmatrix}&-1\\1\end{smallmatrix})\in\spgroup(V,\omega)$.
\end{example}

\section{Conclusion}

This paper sets out to explore encodings of projective Cliffords $\gamma$.
Encoding $\gamma$ as a pair of functions $(r,\psi)$ where $\gamma(D_v)=\zeta^{rv}D_{\psi v}$ is problematic because it is unclear which functions $r$ are allowed. Extended $\Delta$-encodings $(\lambda,\phi)$, where $\gamma(\Delta_v)=\zeta^{\lambda v}\Delta_{\phi v}$, have additional drawbacks, such as non-uniqueness. We saw that these problems are solved by using condensed encodings $(\mu,\psi)$, where for basis vectors $b\in V$, one has $\gamma(\Delta_{\underline b})=\zeta^{\mu b}\Delta_{\underline{\phi b}}$. The formula for $\gamma(\Delta_v)$ is more complicated, but we derive it and related formulas. Condensed encodings have applications, particularly where linearity can be exploited, as seen in in~\citep{paykin2024qudit}.

\begin{acknowledgments}
Many thanks to Albert Schmitz and Jon Yard for numerous fruitful discussions about this work.
The authors would also like to thank the anonymous reviewers for their extremely helpful feedback.
\end{acknowledgments}

\section*{Author Declarations}

The work of Sam Winnick was supported in part by the NSERC Discovery under Grant No. RGPIN-2018-04742 and the NSERC project FoQaCiA under Grant No. ALLRP-569582-21.

Jennifer Paykin reports the following recent sources of financial support: New York University Abu Dhabi. Travel support, 04/24; Perimeter Institute. Travel reimbursement, 04/24; Computing Community Consortium. Travel reimbursement, 05/23.

\appendix

\section{Relations on sets and groups}\label{appendix}
If $X$ and $Y$ are sets, a \textit{relation} from $X$ to $Y$ is a triple $(X,Y,\gamma)$ where $\gamma$ is a subset of their Cartesian product, $\gamma\subseteq X\times Y$. Abusing a notation typically just used for functions, we often write $x\mapsto y$ to mean $(x,y)\in\gamma$. Likewise, if $X$ and $Y$ are groups, then a \textit{homomorphic relation} from $X$ to $Y$ is a triple $(X,Y,\gamma)$ where $\gamma$ is a subgroup of their direct product, $\gamma\leq X\times Y$. In other words, $\gamma$ is a subset of the Cartesian product of the underlying sets, subject to the group axioms: ($i$) $(1,1)\in\gamma$; ($ii$) if $(x,y)\in\gamma$ then $(x^{-1},y^{-1})\in\gamma$; and ($iii$) if $(x_1,y_1)\in\gamma$ and $(x_2,y_2)\in\gamma$ then $(x_1x_2,y_1y_2)\in\gamma$. A (homomorphic) relation $\gamma$ from $X$ to $Y$ is \textit{left-total} if for all $x\in X$ there exists $y\in Y$ such that $(x,y)\in\gamma$. \textit{Right totality} (also called \textit{surjectivity}) is defined dually. A (homomorphic) relation $\gamma$ from $X$ to $Y$ is \textit{right-definite} if whenever $(x,y_1),(x,y_2)\in\gamma$, it follows that $y_1=y_2$. \textit{Left definiteness} (also called \textit{injectivity}) is defined dually. A \textit{function} is a left-total right-definite relation, and a \textit{homomorphism} is a left-total right-definite homomorphic relation. The \textit{kernel} of a homomorphic relation $\gamma$ from $X$ to $Y$ is $\mathrm{ker\,}\gamma=\{x\in X\,|\,(x,1)\in\gamma\}\subseteq X$. 

\begin{lemma}\label{lem-kersubgroup}
    Let $X$ and $Y$ be groups and let $\gamma:X\to Y$ be a homomorphic relation. Then $\mathrm{ker\,}\gamma\leq X$.
\end{lemma}

\begin{proof}
    By ($i$) above, $1\in\mathrm{ker\,}\gamma$. Now let $x\in\mathrm{ker\,}\gamma$. Then by ($ii$) above, $x^{-1}\in\mathrm{ker\,}\gamma$. Finally let $x_1,x_2\in\mathrm{ker\,}\gamma$. Then by ($iii$) above, $(x_1,1)(x_2,1)=(x_1x_2,1)$, so $x_1x_2\in\mathrm{ker\,}\gamma$. Therefore $\mathrm{ker\,}\gamma$ is a subgroup, as claimed.
\end{proof}

\begin{lemma}\label{lem-injlem}
    Let $X$ and $Y$ be groups and let $\gamma:X\to Y$ be a homomorphic relation. Then $\gamma$ is left-definite (injective) if and only if $\mathrm{ker\,}\gamma$ is trivial.
\end{lemma}

\begin{proof}
    $(\implies)$: Let $x\in\mathrm{ker\,}\gamma$. By Lemma (\ref{lem-kersubgroup}) and property ($i$) above, $1\in\mathrm{ker\,}\gamma$, so $(x,1),(1,1)\in\gamma$, so $x=1$. $(\impliedby)$: Let $(x_1,y),(x_2,y)\in\phi$. Then by Lemma (\ref{lem-kersubgroup}) and properties ($ii$) and ($iii$) above, $(x_1^{-1},y^{-1})\in\gamma$, and then $(x_1^{-1}x_2,1)\in\gamma$, so $x_1^{-1}x_2\in\mathrm{ker\,}\gamma$, so $x_1=x_2$.
\end{proof}

\bibliography{bibliography}

\end{document}